\documentclass[copyright,creativecommons]{eptcs}
\providecommand{\event}{ICLP 2026}

\usepackage{iftex}
\ifpdf
  \usepackage{underscore}
  \usepackage[T1]{fontenc}
\else
  \usepackage{breakurl}
\fi

\usepackage{amsmath,amssymb,amsthm}
\usepackage{booktabs}
\usepackage{graphicx}
\usepackage{listings}
\usepackage{xcolor}
\usepackage{enumitem}
\usepackage{microtype}
\usepackage{needspace}
\usepackage{tikz}
\usetikzlibrary{arrows.meta,positioning,fit,backgrounds}

\newtheorem{definition}{Definition}
\newtheorem{proposition}{Proposition}

\newcommand{\asp}{\textnormal{\textsc{ASP}}}
\newcommand{\eps}{\textnormal{\textsc{EPS}}}
\newcommand{\proj}[1]{\mathcal{P}_{#1}}

\lstdefinelanguage{ASP}{
  morekeywords={not},
  sensitive=true,
  alsoletter={\#},
  morecomment=[l]{\%},
  morestring=[b]",
}

\title{Explainable Belief Harmonization\\under Dynamic Epistemic Partitions}
\author{Adam Kostka
\institute{Warsaw University of Technology, Warsaw, Poland}
\email{adam.kostka.stud@pw.edu.pl}
\and
Jaros{\l}aw A. Chudziak
\institute{Warsaw University of Technology, Warsaw, Poland}
\email{jaroslaw.chudziak@pw.edu.pl}}
\def\titlerunning{Explainable Belief Harmonization under Dynamic Epistemic Partitions}
\def\authorrunning{Kostka \& Chudziak}

\begin{document}
\maketitle

\begin{abstract}
Existing approaches to multi-agent belief combination have established mature foundations for combining uncertain beliefs under common assumptions: consensus methods use iterative averaging, logic-based methods resolve conflicting knowledge bases, and epistemic logic analyzes agents' information states. Typically, these approaches assume that the structure determining what each agent can represent remains fixed. However, in many scenarios, agents gain or lose observational capacity during execution, and what was once admissible may become structurally impossible. This paper presents a formal framework for handling such runtime changes in epistemic partitions over continuous belief profiles. A hybrid approach exploits the advantages of answer set programming in elaboration tolerance, declarative integrity constraints, and explanations, with the numerical flexibility of Python. The framework applies to domains where agents operate at heterogeneous and possibly changing levels of resolution, and provides formal guarantees of admissibility preservation under refinement, unique mass-preserving repair under coarsening, and explanation completeness. Evaluation across 100 randomly generated topology changes confirms complete violation detection and explanation coverage.
\end{abstract}

\section{Introduction}\label{sec:introduction}

Combining uncertain beliefs for multiple agents is an established problem with mature theories for its solution, including consensus theory \cite{degroot1974consensus,hegselmann2002opinion}, logic-based merging \cite{konieczny2002merging,konieczny2011logicbasedmerging}, and epistemic logic \cite{fagin2004reasoning}. These theories provide effective solutions for averaging over multiple iterations, merging conflicting knowledge bases, and logical analysis of agents' information states. In aggregate, they provide a solid foundation for combining uncertain beliefs for multiple agents when the representational structure is known a priori.

However, in many practical cases, this representational structure is dynamic. Agents gain or lose observational capacity, new agents can be added, or existing agents can be removed, and previously admissible beliefs can become structurally inconsistent with respect to the new structure \cite{fagin2004reasoning}. The system must recognize this inconsistency, identify its cause, correct it, and continue. Dynamic representational structures occur naturally in systems where agents have different and changing levels of resolution.

For example, imagine a system where multiple sensors provide probability distributions over a set of conditions, but at different levels of resolution, which can be dynamic. A high-resolution LiDAR sensor can distinguish ten different surfaces, a camera can aggregate them at five levels, and an ultrasonic sensor can distinguish between obstacles and no obstacles. If it rains and the camera's resolution changes or a sensor fails, the representational structure changes, and previously admissible beliefs may no longer be valid. The existing theories for combining uncertain beliefs assume fixed representational structures, while logic-based merging assumes discrete knowledge bases. The nearest existing solution for combining uncertain beliefs is Equibel \cite{vicol2015equibel}, where discrete propositional beliefs are combined on a graph topology. The problem is how to handle dynamic changes to a continuous belief structure and automatically explain the changes.

In order to address this problem, the study sets out to examine how a system can handle changes to agents' observational capabilities during belief harmonization, while maintaining consistency and explaining what changed. The main argument is that \asp{} is able to provide an abstraction boundary for the discrete structural part of the epistemic constraints. \asp{} is also able to provide a declarative structural controller for reconfiguration, integrity, violation localization, and explanation, while Python is used for the numerical part. Three hypotheses are developed to support this argument. The first hypothesis (H1) is that the management of partition changes is possible using \asp{} facts without modifying the rules, thus capturing the elaboration tolerance principle \cite{mccarthy1998elaboration} in a novel domain. The second hypothesis (H2) is that \asp{} integrity constraints are used to provide a declarative solution to the well-formedness of the structure, thus capturing malformed edits. The third hypothesis (H3) is that the solver is able to automatically provide explanation atoms that are precise with regard to the distinctions and violations.

\begin{figure}[t]
\centering
\begin{tikzpicture}[
    scale=0.85, transform shape,
    venn circle/.style={draw=black!70, circle, minimum size=6.0cm, fill=black!10, fill opacity=0.4, text opacity=1, thick},
    node distance=0mm,
    font=\small,
    align=center
]
    \coordinate (C1) at (90:2.0);
    \coordinate (C2) at (210:2.0);
    \coordinate (C3) at (330:2.0);

    \node[venn circle] at (C1) {};
    \node[venn circle] at (C2) {};
    \node[venn circle] at (C3) {};

    \node[font=\small\bfseries, text=black!90] at (0, 3.7) {Logic \& Epistemic\\Structure};
    \node[font=\small\bfseries, text=black!90] at (-2.8, -2.4) {Continuous Beliefs\\\& Consensus};
    \node[font=\small\bfseries, text=black!90] at (2.8, -2.4) {Hybrid\\Architectures};

    \node at (-1.2, 2.5) {DEL\\\cite{vanditmarsch2007del}};
    \node at (1.2, 2.5) {Konieczny\\\cite{konieczny2002merging}};
    
    \node at (-2.8, -0.8) {DeGroot\\\cite{degroot1974consensus}};
    
    \node at (2.8, -0.8) {Equibel\\\cite{vicol2015equibel}};
    
    \node at (0, -1.8) {D'Asaro\\\cite{dasaro2024epec}};

    \node[font=\small\bfseries, fill=white, fill opacity=0.9, text opacity=1, draw=black!60, rounded corners, inner sep=4pt] at (0, 0) {This work};

\end{tikzpicture}
\caption{Positioning relative to nearest approaches. Three overlapping areas represent epistemic structure and logic programming, continuous beliefs and consensus, and hybrid solver--driver architectures. This work sits at the intersection of all three.}
\label{fig:comparison}
\end{figure}
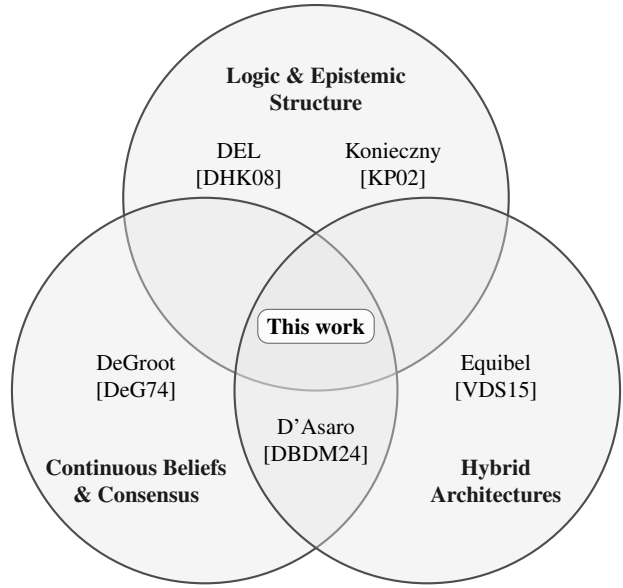

The contribution presented in this paper proposes a framework in which changes in partitions are treated as main events during runtime, and such changes can invalidate belief states previously deemed admissible. Although each concept, such as epistemic partitions, projection repair, hybrid \asp{}, and multi-shot solving, has been previously studied, the contribution is how these concepts are combined to facilitate the automatic detection and handling of changes in domains where agents operate with different and varying levels of resolution, such as sensor fusion, federated analytics, and collaborative diagnosis. Concretely, the paper contributes (i)~a formal model of dynamic epistemic partition structures, (ii)~a declarative \asp{} encoding supporting fact-only reconfiguration, violation localization, and explanation generation, and (iii)~an experimental evaluation confirming the three hypotheses.

The remainder of this paper is organized as follows: Section \ref{sec:related} discusses related work, Section \ref{sec:model} formalizes the problem, Section \ref{sec:encoding} describes the implementation, Section \ref{sec:evaluation} presents the empirical evaluation, Section \ref{sec:discussion} discusses scope and limitations, and Section \ref{sec:conclusion} concludes.

\section{Related Work}\label{sec:related}

The problem of achieving consistency among agents with differing observational capabilities is a problem at the intersection of several traditional research areas \cite{fagin2004reasoning,degroot1974consensus}. Figure~\ref{fig:comparison} illustrates how these areas overlap. There are three threads in the discussion: traditions of logic programming with regard to epistemic structure and uncertainty, hybrid architectures of solver and driver, and traditions of consensus and belief revision.

\subsection{Epistemic and Probabilistic Logic Programming}

The tradition of epistemic specifications starts with the idea of strong introspection due to Gelfond \cite{gelfond1991strong}, classical negation due to Gelfond and Lifschitz \cite{gelfondlifschitz1991classical}, then epistemic negation due to Shen and Eiter \cite{sheneiter2016epistemic}, autoepistemic expansions due to Cabalar et al.~\cite{cabalar2020autoepistemic}, properties of epistemic logic programs due to Costantini and Formisano \cite{costantini2024epistemic}, and more recent refinements due to Su et al.~\cite{su2025pearce}. Epistemic planning in multi-agent systems with the \asp{} paradigm investigates Kripke structures and sequences of actions due to Burigana and Fabiano \cite{burigana2020planning,fabiano2021lies}. All of these traditions demonstrate the potential of logic programming to make epistemic structure and uncertainty explicit, but they focus on qualitative reasoning in the logical framework and do not deal with quantitative probabilistic reasoning constrained by the representational abilities of the agents that may change during execution.

Probabilistic logic programming approaches the problem of uncertainty from a different angle. This tradition starts with probabilistic Horn abduction due to Poole \cite{poole1993horn}, continues with distribution semantics due to Sato \cite{sato1995distribution}, and then with the system of ProbLog and its weighted variant due to De Raedt et al.~\cite{deraedt2007problog,fierens2015weighted}. Answer set programming approaches include P-log due to Baral et al.~\cite{baral2009probabilistic}, LPMLN due to Lee et al.~\cite{lee2017lpmln}, and the recent Plingo system \cite{hahn2025plingo}. These systems reason probabilistically inside the logic itself. The design point of the present work is narrower: probability belongs to the numerical layer, while \asp{} represents structural admissibility and generates explanations. Neither tradition addresses what happens when the epistemic structure itself changes during execution.

\subsection{Hybrid Solver-Driver Architectures}

The most similar methodological predecessors are hybrid systems, which include declarative reasoning in the most expressive domain while moving other computations to an external layer. Research into the integration of \asp{} and description logics for the Semantic Web initiated the paradigm of ``a solver coupled with an external theory'' \cite{eiter2008semanticweb}. Subsequently, the paradigm of ``multi-shot \texttt{clingo} control'' has been proposed, which allows the solver state to be preserved across iterative solving processes \cite{gebser2019multishot}. In the area of language interoperability, Natlog has been proposed as an integration of logic programming within the Python deep learning stack \cite{tarau2023natlog}, while Janus is a system for multi-paradigm Prolog+Python programming \cite{swift2023janus}. Specifically, D'Asaro et al.'s \asp{}-based approach to epistemic probabilistic event calculus \cite{dasaro2024epec} combines \asp{} with probabilistic epistemic reasoning in a hybrid architecture, differing from the current approach in the focus on temporal event sequences rather than dynamic partition structures.

The above-mentioned predecessors confirm the applicability of the engineering pattern of ``a solver coupled with an external theory.'' In the current research, the pattern is followed, while the solver is \asp{}, and the external theory is Python, which is responsible for numerical operations involving belief structures. The approach is applied to the new domain of constrained probabilistic harmonization under various epistemic partitions, which has not been covered in the above-mentioned predecessors.

\subsection{Consensus, Belief Change, and Structural Dynamics}

Besides logic programming, there are other areas of study concerning consensus and merging of beliefs. DeGroot's model deals with averaging-based consensus under fixed interaction structures, where it is assumed that all agents have the same representational capacity \cite{degroot1974consensus}. Hegselmann and Krause's bounded confidence model restricts who an agent can listen to, leading to clustering instead of global consensus \cite{hegselmann2002opinion}. These models assume all agents have the same representational capacity; however, having partitions restrict what an agent can represent is a different kind of restriction. At a more basic level, Aumann's Agreement Theorem proves that Bayesian agents with a common prior cannot agree to disagree if their posteriors become common knowledge \cite{aumann1976agreeing}. Agents with different partitions do not fit this model because they cannot represent the same posterior, thus disagreeing is a structural, not irrational, property.

Logic-based belief merging using distance-based operators is applicable to the integration of conflicting propositional knowledge bases with explicit integrity constraints \cite{konieczny2002merging,konieczny2011logicbasedmerging}, and this is well understood, although not in the continuous domain with agent-specific partitions. The Equibel system \cite{vicol2015equibel} is an implementation of consistency-based multi-agent belief change, using an \asp{} + Python workflow on graph topologies, and this is very close architecturally to the system presented here. The same type of hybrid symbolic-agent system is found in multi-agent LLM systems that have a logical layer for knowledge management and coordination \cite{kostka2024synergizing}. The main difference is the use of the solver, as \asp{} is used to compute the belief change in the Equibel system, while Python is used in the system presented here. Dynamic Epistemic Logic \cite{vanditmarsch2007del} is another approach to the problem, using modal logic to reason about the change in the knowledge, although without the use of \asp{} or continuous probability. Elaboration tolerance, the ability to adapt to new pieces of information by editing the facts rather than the programs, is the fundamental requirement driving knowledge representation \cite{mccarthy1998elaboration}, and this is achieved using \asp{} through the declarative separation of facts and rules. The need to provide explanations for the reasoning performed by the solver, as indicated by the recent appearance of explanation techniques for \asp{} programs \cite{fandinno2019xasp,alviano2024xasp2}, suggests that the problem is gaining popularity, and the current traditions do not handle the problem as presented, although they are related and address other aspects of the problem.

\section{Problem Formulation}\label{sec:model}

The preceding discussion revealed a gap at the intersection of epistemic structure, dynamic change, and continuous beliefs. This section formalizes this problem, including the structural objects, the changes that are made to these objects, and the repair that is performed to ensure consistency. The refinement and coarsening propositions help to explain why dynamic change is asymmetric. The projection uniqueness theorem justifies the repair operator as non-ad hoc. The explanation completeness theorem ties changes to symbolic beliefs to actual admissibility violations.

\subsection{Problem Statement}

The fundamental question that this research seeks to address is how a system can handle changes to agents' observational capabilities during belief harmonization, while maintaining consistency and explaining what changed.

Let $W$ be a finite set of possible worlds and $A$ a finite set of agents. Each agent $a \in A$ observes the world through an epistemic partition $\Pi_a$ of $W$ into cells \cite{fagin2004reasoning}. Those possible worlds within a cell are indistinguishable to agent $a$. Each agent $a \in A$ also has a probability distribution $\mu_a$ over $W$, representing its beliefs about the actual world. The partition imposes a constraint on the representational capabilities of the agent. The agent is unable to represent different probabilities for worlds that are indistinguishable to itself.

This constraint is the source of the problem. As partitions change over time, beliefs that are consistent with the structural constraints may become invalid.

During execution, various types of changes can occur. The most common are changes to an agent's partition. A cell can either refine or coarsen. Agents can also be added or deleted. With regard to the specific domain, the set of possible worlds $W$ is usually finite. More formally, with a set of agents with varying partitions and a belief profile $\mu = \{\mu_a\}_{a \in A}$, the goal is to satisfy the following three objectives:
\begin{enumerate}
\item Determine whether $\mu$ is structurally consistent with the current partition configuration.
\item If not, compute a repaired profile that restores consistency while staying close to the original beliefs.
\item Generate an explanation that identifies exactly which distinctions were gained or lost.
\end{enumerate}

\subsection{Partitions and Admissibility}

Returning to the sensor fusion scenario from the introduction, let
\[
W = \{w_0,\; w_1,\; w_2,\; w_3,\; w_4,\; w_5\}
\]
represent six surface conditions. The LiDAR sees three groups:
\[
\Pi_{\text{lidar}} = \{\{w_0,w_1,w_2\}, \{w_3,w_4\}, \{w_5\}\}.
\]
The camera has access to two groups: $\Pi_{\text{cam}} = \{\{w_0, w_1, w_2\}, \{w_3, w_4, w_5\}\}$. The ultrasonic sensor cannot distinguish any of the states. Therefore, $\Pi_{\text{ultra}} = \{W\}$. This gives us three sensors with varying levels of granularity over the same reality. In terms of \asp{}, this equates to \texttt{s1} (LiDAR), \texttt{s2} (camera), and \texttt{s3} (ultrasonic).

\begin{definition}[Epistemic Partition Structure]
\label{def:eps}
An epistemic partition structure (\eps{}) is a tuple $\mathcal{E} = \langle W, A, \{\Pi_a\}_{a \in A}\rangle$ where $W$ is a finite set of worlds, $A$ is a finite set of agents, and each $\Pi_a$ is a partition of $W$ into cells.
\end{definition}

\begin{definition}[Admissible belief profile]
\label{def:admissibility}
A belief profile $\mu = \{\mu_a\}_{a \in A}$, where each $\mu_a$ is a probability distribution over $W$, is admissible with respect to \eps{} $\mathcal{E}$ if for every agent $a \in A$ and every cell $C \in \Pi_a$: $\forall w, w' \in C\colon \mu_a(w) = \mu_a(w')$.
\end{definition}

This definition expresses an important and basic idea. An agent must have the same probabilities for worlds it cannot distinguish. If the camera cannot distinguish $w_3$, $w_4$, and $w_5$ because they are in the same cell, it must give each of these worlds the same probability. Singleton cells vacuously satisfy admissibility because the universal equality condition is over one world. The admissibility constraint is both semantic and runtime.

\subsection{Topology Operations}

In practice, the structure of the partitions need not be fixed, as agents may gain or lose the capacity to observe, and agents may be added to or removed from the system. Dealing with these issues requires elaboration tolerance \cite{mccarthy1998elaboration}. The two fundamental operations on partitions are defined as follows.

\begin{definition}[Refinement and coarsening]
\label{def:refinement}
Partition $\Pi'_a$ refines $\Pi_a$ (written $\Pi'_a \leq \Pi_a$) if every cell of $\Pi'_a$ is a subset of some cell of $\Pi_a$. Partition $\Pi'_a$ coarsens $\Pi_a$ if $\Pi_a$ refines $\Pi'_a$.
\end{definition}

Refinement corresponds to an agent gaining observational capacity. In the sensor fusion scenario, when a firmware upgrade improves the camera's resolution, it can now distinguish $w_0$ and $w_1$ from $w_2$. Its cell $\{w_0, w_1, w_2\}$ splits into $\{w_0, w_1\}$ and $\{w_2\}$. Coarsening corresponds to losing capacity. When rain degrades the LiDAR's resolution, its cells merge.

The two operations have opposite effects on admissibility.

\begin{proposition}[Refinement preserves admissibility]
\label{prop:refinement}
If $\mu_a$ is admissible with respect to $\Pi_a$ and $\Pi'_a$ refines $\Pi_a$, then $\mu_a$ is admissible with respect to $\Pi'_a$.
\end{proposition}

\begin{proof}
Let $C' \in \Pi'_a$. Since $\Pi'_a$ refines $\Pi_a$, there exists $C \in \Pi_a$ with $C' \subseteq C$. Since $\mu_a$ is admissible with respect to $\Pi_a$, $\mu_a$ is uniform on $C$. Since $C' \subseteq C$, $\mu_a$ is uniform on $C'$.
\end{proof}

\noindent This means refinement events (the most common topology change in practice) require no repair step, reducing the runtime cost of gaining observational capacity to a single re-solve.

\begin{proposition}[Coarsening can violate admissibility]
\label{prop:coarsening}
There exist $\mu_a$ admissible with respect to $\Pi_a$ and $\Pi'_a$ coarsening $\Pi_a$ such that $\mu_a$ is not admissible with respect to $\Pi'_a$.
\end{proposition}

\begin{proof}
Let $W = \{w_1, w_2\}$, $\Pi_a = \{\{w_1\}, \{w_2\}\}$, $\mu_a(w_1) = 0.7$, $\mu_a(w_2) = 0.3$. This is admissible (singleton cells). Coarsen to $\Pi'_a = \{\{w_1, w_2\}\}$. Now $\mu_a(w_1) \neq \mu_a(w_2)$ within the same cell.
\end{proof}

\subsection{Admissibility Repair}

When coarsening violates admissibility, beliefs must be repaired. The natural repair operation averages beliefs within each cell of the new partition.

\begin{definition}[Projection]
\label{def:projection}
Given a distribution $\mu_a$ over $W$ and a partition $\Pi_a$, the projection $\proj{a}(\mu_a)$ assigns to each $w \in C \in \Pi_a$ the cell average:
\[
\proj{a}(\mu_a)(w) = \frac{1}{|C|}\sum_{w' \in C} \mu_a(w').
\]
\end{definition}

By construction, $\proj{a}(\mu_a)$ is admissible with respect to $\Pi_a$. The following proposition shows that this repair is not merely convenient but uniquely determined.

\begin{proposition}[Projection is the unique mass-preserving admissible repair]
\label{prop:projection}
Given $\mu_a$ over $W$ and partition $\Pi_a$, the projection $\proj{a}(\mu_a)$ is the unique admissible distribution that preserves cell masses: for every $C \in \Pi_a$, $\sum_{w \in C} \proj{a}(\mu_a)(w) = \sum_{w \in C} \mu_a(w)$. It is also the unique $L_2$-minimizer among all admissible distributions over~$W$.
\end{proposition}

\begin{proof}
Uniqueness under mass preservation.
Since admissibility implies $\proj{a}(\mu_a)(w) = c_C$ for all $w \in C$, mass conservation implies that
$c_C \cdot |C| = \sum_{w \in C} \mu_a(w)$.
This implies that $c_C = \frac{1}{|C|} \sum_{w \in C} \mu_a(w)$ is uniquely determined.

$L_2$ Optimality.
Define $V_a = \{ v \in \mathbb{R}^{|W|} \; : \; v \text{ is constant on each cell of } \Pi_a \}$. $V_a$ is a linear subspace of $\mathbb{R}^{|W|}$. By the orthogonal projection theorem, $\proj{a}(\mu_a)$ is the unique minimizer of $\| \mu_a - v \|_2$ over $v \in V_a$. Since $\mu_a$ is a probability distribution, the averages over each cell are nonnegative, and the total mass is preserved. Therefore, $\proj{a}(\mu_a) \in \Delta_{|W| - 1}$. The unconstrained minimizer already satisfies the simplex constraints, so it is also the unique minimizer over $V_a \cap \Delta_{|W|-1}$, the set of all admissible distributions.
\end{proof}

In this sensor fusion problem, when the camera's two cells with belief levels 0.180 and 0.153 combine to form a single cell, projection computes an average of all six worlds at 0.167. This uniqueness implies that no choice or configuration is required for the repair operator; it is determined purely based on the partition.

Having described this formal framework, the next section will show how this is achieved using \asp{}.

\section{Model Realization}\label{sec:encoding}

The structural layer directly corresponds to \asp{} facts and rules. The encoding has been designed to be concise, with just fifteen rules needed to cover any number of agents, any number of worlds, and any possible change in topology. To illustrate, we present a listing here in which we constrain cell identities to \texttt{C=1..10}, although in practice, we dynamically determine the appropriate range from the provided partition. The rules, however, are invariant.

\subsection{Partition Management}

The core encoding uses external facts to represent partition membership and derives indistinguishability directly.

\begin{lstlisting}[caption={Core \asp{} encoding for dynamic epistemic structure.},label={lst:core}]
world(w0). world(w1). world(w2). world(w3). world(w4). world(w5).
agent(s1). agent(s2). agent(s3). agent(critic).

% Partition cells -- external (changed dynamically)
#external partition(A, C, W) : agent(A), world(W), C = 1..10.

% Indistinguishability derived from shared cell membership
indist(A, W1, W2) :- partition(A, C, W1), partition(A, C, W2), W1 != W2.

% Well-formedness: every agent covers every world
covered(A, W) :- partition(A, _, W).
:- agent(A), world(W), not covered(A, W).

% Well-formedness: no world in two cells
:- partition(A, C1, W), partition(A, C2, W), C1 != C2.
\end{lstlisting}

The \texttt{\#external} declaration is the key enabler. It tells \texttt{clingo} that \texttt{partition/3} atoms may be asserted or retracted at any time without re-grounding the program \cite{gebser2019multishot}. A topology change reduces to a sequence of \texttt{assign\_external} calls in Python: when the camera's cell $\{w_0, w_1, w_2\}$ splits into $\{w_0, w_1\}$ and $\{w_2\}$, three atoms are retracted and three asserted. Zero rule changes. The solver automatically recomputes \texttt{indist/3} atoms under the new partition. This is elaboration tolerance \cite{mccarthy1998elaboration} in action.

\begin{proposition}[Elaboration tolerance]
\label{prop:elaboration}
Any reconfiguration of an \eps{} $\mathcal{E}$ to a new structure $\mathcal{E}'$ over the same world and agent sets is realizable by modifying only the external atoms in \texttt{partition/3}, without modifying the rules of the \asp{} encoding.
\end{proposition}

The encoding is doing something more than just storing the data; it is a declarative structural controller. The well-formedness constraints filter out invalid partitions at solve time, the \texttt{indist/3} atoms generated determine the structure for admissibility, and the \texttt{violates/3} atoms determine where particular admissibility violations are, and the explanation rules calculate relation deltas that would otherwise require reimplementing the indistinguishability and violation logic procedurally in the numerical layer. The resulting proposition directly relates to the formal framework, showing that the formal definitions are implemented by the encoding.

\begin{proposition}[Encoding correctness]
\label{prop:encoding}
For any valid \eps{} input (each agent's partition covers every world with no overlapping cells), the \asp{} program in Listings~\ref{lst:core} and~\ref{lst:explain} has a unique stable model in which: (i)~$\texttt{indist}(a, w_1, w_2)$ holds iff $w_1$ and $w_2$ belong to the same cell in $\Pi_a$; (ii)~$\texttt{violates}(a, w_1, w_2)$ holds iff $w_1 \sim_a w_2$ and $\texttt{neq}(a,w_1,w_2)$ is asserted.
\end{proposition}

\begin{proof}
The model is stratified, with the external world represented by the atoms \texttt{partition/3} and \texttt{neq/3}, while the derived atoms \texttt{indist/3}, \texttt{covered/2}, and \texttt{violates/3} are computed using Horn clauses. The two integrity constraints ensure that the model is not empty or contains overlapping regions. The unique stable model for valid input will have \texttt{indist}$(a, w_1, w_2)$ if and only if there is some cell $C$ that contains $w_1$ and $w_2$, and will have \texttt{violates}$(a, w_1, w_2)$ if and only if $w_1$ and $w_2$ are also marked as unequal by the Python code.
\end{proof}

\subsection{Hybrid Architecture}
\begin{figure}[t]
\centering
\includegraphics[width=0.65\linewidth]{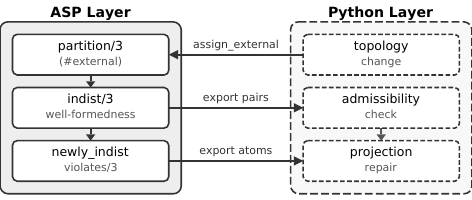}
\caption{Hybrid architecture. The \asp{} layer (left, solid border) manages partition structure, derives explanations, and localizes violations. The Python layer (right, dashed border) handles topology edits, numerical inequality computation, and belief repair.}
\label{fig:arch}
\end{figure}

The encoding scheme presented above is responsible for the structural layer. The current section is responsible for explaining the integration with the numerical layer. Figure~\ref{fig:arch} shows the division of labor. The \asp{} part is responsible for the partition structure, derivation of indistinguishability, well-formedness, violation localization, and explanation generation. Python is responsible for the belief storage, floating-point calculations, and re-projection repairs. Moreover, Python is responsible for calculating the numerical inequality facts (represented as \texttt{neq/3} facts that specify the world pairs with differing beliefs) and passing them to the \asp{} part, which calculates the structural illicitness using the \texttt{violates/3} predicate.

This division separates concerns by design: \asp{} represents the structural constraints and explanations, where elaboration tolerance and declarative violation localization are natural, while Python handles continuous probabilities and projection repair, where floating-point arithmetic is natural.

The multi-shot solving approach is based on the ground-and-edit approach, where the \asp{} program is grounded once, and the \texttt{\#external} atoms are edited and the solving is repeated as the topology changes. There is no need to re-ground the \asp{} program when the partition is changed. The only case that requires re-grounding is the addition of new worlds (new domain elements), and this is more expensive than editing the partition facts.

\subsection{Explanation Generation}\label{sec:explanation}

With the hybrid pipeline in place, the focus is now on the second main feature of the framework, which is to reveal the changes that take place when the topology is reconfigured. This is in line with the general goal of making the reasoning performed by \asp{} transparent \cite{fandinno2019xasp,alviano2024xasp2}. When the topology is modified, the system records the previous partition as \texttt{old\_partition/3} atoms and constructs explanation atoms using Negation as Failure.

\needspace{16\baselineskip}
\begin{lstlisting}[caption={Explanation and violation localization rules.},label={lst:explain}]
#external old_partition(A, C, W) : agent(A), world(W), C = 1..10.
old_indist(A, W1, W2) :- old_partition(A, C, W1), old_partition(A, C, W2),
                          W1 != W2.

newly_dist(A, W1, W2) :- old_indist(A, W1, W2), not indist(A, W1, W2),
                          W1 < W2.
newly_indist(A, W1, W2) :- indist(A, W1, W2), not old_indist(A, W1, W2),
                            W1 < W2.

% Violation localization
#external neq(A,W1,W2) : agent(A), world(W1), world(W2), W1<W2.
violates(A,W1,W2) :- indist(A,W1,W2), neq(A,W1,W2), W1<W2.
\end{lstlisting}

The atom \texttt{newly\_dist(A, W1, W2)} expresses that worlds $W_1$ and $W_2$ were indistinguishable for agent $A$ under the previous topology but are distinguishable under the new one. Conversely, the atom \texttt{newly\_indist(A, W1, W2)} expresses that two worlds previously distinguishable for agent $A$ have become indistinguishable. The rule \texttt{violates/3} closes this loop. Python will derive \texttt{neq/3} facts for pairs of worlds whose beliefs are different, and \asp{} is responsible for deriving \texttt{violates(A, W1, W2)} for those pairs that are supposed to be indistinguishable but are not.

\begin{proposition}[Explanation completeness]
\label{prop:explanation}
Given old \eps{} $\mathcal{E}$ and new \eps{} $\mathcal{E}'$, the explanation atoms $\{\texttt{newly\_dist}(a,w,w'), \texttt{newly\_indist}(a,w,w')\}$ capture exactly the symmetric difference of the indistinguishability relations ${\sim_a} \mathbin{\triangle} {\sim'_a}$. Furthermore, if the belief profile $\mu$ is admissible with respect to $\mathcal{E}$, then any admissibility violation of $\mu$ under $\mathcal{E}'$ must involve a pair in \texttt{newly\_indist}.
\end{proposition}

\begin{proof}
Define $w \sim_a w'$ iff $\exists C \in \Pi_a\colon w, w' \in C$. Then \texttt{newly\_dist}$(a,w,w')$ holds iff $w \sim_a w' \wedge \neg(w \sim'_a w')$, and \texttt{newly\_indist}$(a,w,w')$ holds iff $\neg(w \sim_a w') \wedge w \sim'_a w'$. These are exactly the two halves of ${\sim_a} \mathbin{\triangle} {\sim'_a}$. For the second claim: a violation requires $\mu_a(w) \neq \mu_a(w')$ with $w \sim'_a w'$. If $w \sim_a w'$ already held, admissibility of $\mu$ with respect to $\mathcal{E}$ would have required $\mu_a(w) = \mu_a(w')$, contradicting the violation. So the violating pair must satisfy $\neg(w \sim_a w') \wedge w \sim'_a w'$, which is exactly \texttt{newly\_indist}. The pre-admissibility assumption is necessary: without it, a pair already indistinguishable under $\mathcal{E}$ could already violate uniformity, and such a violation would persist without involving \texttt{newly\_indist}.
\end{proof}

\begin{figure}[t]
\centering
\includegraphics[width=0.88\linewidth]{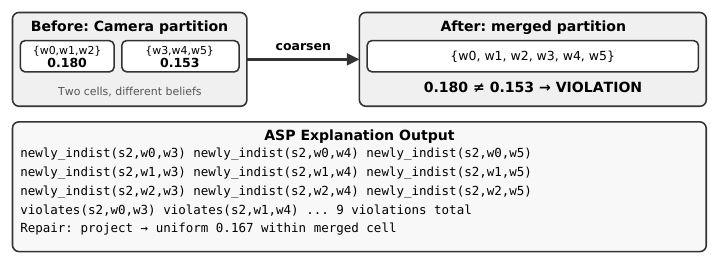}
\caption{Topology change and explanation generation. Left: the camera distinguishes two groups of surface conditions. Right: after coarsening, all conditions merge. The \asp{} layer derives \texttt{newly\_indist} pairs and \texttt{violates} atoms. Projection repair restores uniform beliefs.}
\label{fig:explanation}
\end{figure}

In the sensor fusion scenario, after the refinement done by the LiDAR, the solver asserts \texttt{newly\_dist(s1, w0, w2)} and \texttt{newly\_dist(s1, w1, w2)}, reflecting the fact that world $w_2$ is now distinguishable from worlds $w_0$ and $w_1$. Then, after the coarsening done by the camera, the solver asserts nine \texttt{newly\_indist} statements, each involving the pairing of one world from the first group ($w_0$--$w_2$) and one world from the second group ($w_3$--$w_5$). Figure~\ref{fig:explanation} illustrates this process. The explanation atoms immediately identify the basic cause of the admissibility violation, and the completeness guarantee ensures that no violation escapes the explanation mechanism.

With the encoding in place, the next section evaluates whether the formal properties hold experimentally.

\needspace{6\baselineskip}
\section{Evaluation}\label{sec:evaluation}

The evaluation process considers the three aspects of the declarative structural control thesis: reconfiguring via fact edits (H1), structural integrity (H2), and the quality of the explanations (H3). The sensor fusion scenario is a good example of this; the framework applies to domains matching Definition~\ref{def:eps}, and the random topology experiments conducted in Table~\ref{tab:scale} and Section~5.3 directly test the general formal object. All experiments use the multi-shot API of the \texttt{clingo} system~\cite{gebser2019multishot} accessed via the Python interface. All timings are the median of five runs and can be found in the repository accompanying this paper.

\subsection{Reconfiguration Stress Test}

H1, elaboration tolerance, is evaluated by performing 50 consecutive topology operations on a $20 \times 8$ configuration, with 20 worlds and 8--12 agents, including refinement, coarsening, adding and deleting agents, and a variety of intentionally erroneous edits. The \asp{} rules are kept unchanged throughout this experiment.

The median time to re-solve is 10.3\,ms, with a maximum time of 180.6\,ms on a cold cache, showing that multi-shot editing through external edits is indeed possible. All eleven intentionally erroneous edits are caught by integrity constraints during solve time, marked as UNSAT. The single coarsening operation leading to violations is caught by \texttt{violates/3} and fixed by re-projection. No changes to the rules are needed throughout this 50-step process.

Table~\ref{tab:scale} shows how grounding and solving scale with problem size. The main bottleneck is grounding \texttt{indist/3} atoms, which scales like $O(|A| \cdot |W|^2)$. The cost of re-solving after a topology change on an already grounded program has a median cost of 10\,ms, showing that our multi-shot approach indeed amortizes this cost effectively.

\begin{table}[t]
\centering
\caption{Scalability of the hybrid \asp{} encoding with random partitions.}
\label{tab:scale}
\small
\begin{tabular}{@{}ccccc@{}}
\toprule
$|W| \times |A|$ & Grounding & Solving & Total & Indist atoms \\
\midrule
$6 \times 4$ & 1\,ms & $<$1\,ms & 2\,ms & 82 \\
$10 \times 6$ & 6\,ms & 4\,ms & 10\,ms & 212 \\
$20 \times 8$ & 97\,ms & 58\,ms & 155\,ms & 888 \\
$50 \times 10$ & 3.2\,s & 1.8\,s & 5.0\,s & 6\,392 \\
$100 \times 15$ & 50.8\,s & 30.7\,s & 81.4\,s & 37\,343 \\
\bottomrule
\end{tabular}
\end{table}

\subsection{Structural Integrity Benchmark}

This evaluation is for H2, declarative integrity enforcement. It introduces three classes of malformed partitions: omission of world coverage, overlaps, and stale facts after agent removal. Each variety is tested 20 times for a $10 \times 4$ configuration, comparing the hybrid system to a Python baseline that uses explicit coverage and overlap validation.

Both systems detect all 60 classes of malformed partitions (100\% detection rate). The difference is qualitative: the \asp{} layer catches all three classes through two declarative integrity constraints already in the encoding, requiring no additional validation code. The Python baseline requires dedicated procedural logic for each class. While Python is faster per check (0.01\,ms vs.\ 3\,ms), the question is whether structural invariants live inside the encoding or alongside it as separate obligations.

\subsection{Explanation Completeness}

This experiment verifies H3 (explanation coverage) and provides a stress test for Proposition~\ref{prop:explanation} by simulating 100 random topology changes in the $12 \times 6$ configuration. Each test begins with an admissible belief profile, applies a random partition change to one agent (which is either refinement, coarsening, or reassignment of a cell that is neither refinement nor coarsening), and checks whether each admissibility violation is caused by a \texttt{newly\_indist} pair.

In the 100 tests, 673 violations are found. All violations (100\%) are caused by world pairs classified as \texttt{newly\_indist}, thus strongly supporting the proposition. All 54 coarsening tests have violations, and all violations are pinpointed by the explanation atoms. The \texttt{violates/3} atoms computed in the \asp{} code pinpoint each violation to the specific agent and world pair, instead of the global ``admissibility failed'' message. Since the tests are performed on randomly generated partitions, not the sensor fusion scenario, they test the general formal object. To confirm domain independence, the same framework was applied to a sensor fusion instance with three sensors over eight road conditions and 20 topology changes. All 12 violations traced to \texttt{newly\_indist} pairs (100\% coverage), three malformed edits were caught, and no rule changes were necessary.

\section{Discussion}\label{sec:discussion}

The benchmarks support the main thesis within the considered scope, as fact-only reconfiguration is feasible, declarative integrity constraints detect malformed edits without the need to add additional validation code, and explanation atoms have complete violation coverage. An ablation study on the sensor fusion scenario confirms that each \asp{} component adds distinct value: without explanation atoms, violations cannot be traced to their structural cause; without multi-shot solving, each topology change requires re-grounding.

A natural question is why the structural layer was not simply replaced with Python. If the partition is fixed, a procedural solution would be faster. The declarative layer provides significant benefit when partitions change during execution, because elaboration tolerance \cite{mccarthy1998elaboration} means only \texttt{partition/3} facts need editing, while a procedural solution would require updating validation logic, cell structures, and explanation code separately \cite{sadowski2025verifiable,sadowski2025explainable,fandinno2019xasp}. The \texttt{violates/3} and explanation atoms together provide a diagnostic chain from structural cause to numerical symptom, consistent with the pattern of explicitly assessing internal belief states in multi-agent systems \cite{kostka2025evaluating}.

The approach has several limitations. The grounding cost of the \texttt{violates/3} predicate has an upper bound of $O(|A| \cdot |W|^2)$, which in practice only allows us to interactively use the approach for approximately 50 worlds, although incremental grounding or agent-specific grounding heuristics could improve the limit. The explanations are diagnostic rather than prescriptive, i.e., the system will indicate what has changed but not how to fix it. Closing this gap is an interesting area of further research, which relates to the concept of abductive reasoning in \asp{}. The evaluation is based on a single domain, which has been artificially scaled, so the general applicability of the approach would be better supported by the extension of the framework to other domains. Other repair operators, e.g., maximum entropy or minimum KL-divergence, are not explored in the current study.

\section{Conclusion}\label{sec:conclusion}

In this paper, a system is presented that can handle changes to agents' observational capabilities during belief harmonization, while maintaining consistency and explaining what changed. The main idea presented here is that epistemic partitions not only determine what agents believe, but also what structural capabilities exist for representation, and how refinement and coarsening differ fundamentally from one another in terms of admissibility. Thus, when partitions change, a structural explanation for what has occurred and why is necessary, rather than remediation in terms of numbers.

In this regard, a formal framework for dynamic epistemic partition structures has been developed and implemented using a hybrid \asp{} + Python approach, in which changes in structure are translated into editing external atoms in \texttt{clingo} \cite{gebser2019multishot}, and explanation atoms are derived automatically by the solver, providing exact identification of what has been gained or lost in terms of distinctions. Six propositions define the formal properties of this framework, and evaluation has provided support for the hypotheses (H1)--(H3).

Three avenues for further research are presented, including providing prescriptive explanations through abductive reasoning, integrating with probabilistic \asp{} systems such as Plingo \cite{hahn2025plingo}, and neuro-symbolic approaches, in which neural beliefs are constrained by the structure enforced by \asp{}. On a broader scale, a design pattern has been developed, in which \asp{} functions as a structural runtime controller for numerical components in a system whose discrete constraints change during execution.

\bibliographystyle{eptcsalpha}
\bibliography{iclp}

\end{document}